\documentclass[11pt]{article}

\usepackage[all]{xy}           
\usepackage{amsmath,amsfonts,amssymb,amsthm,textcomp}
\usepackage{eucal}
\usepackage{epsfig}

\def\v #1{\vert #1\vert}             

\def\m #1 #2{(-1)^{{\v #1} {\v #2}}} 

\theoremstyle{plain}
\newtheorem{theorem}{Theorem}

\theoremstyle{definition}
\newtheorem{definition}[theorem]{Definition}

\def\<#1>{\langle#1\rangle}

\usepackage{epsfig}
\usepackage{amsmath}
\usepackage{epic}
\usepackage{amssymb}
\usepackage{fancybox}
\usepackage{wrapfig}
\usepackage[font=footnotesize,labelsep=period]{caption}

\usepackage{float}

\usepackage{booktabs}

\usepackage[all]{xy}

\usepackage{fancyhdr}

\usepackage{fancyhdr}

\usepackage{appendix}

\begin{document}

\centerline{\Large \bf A geometric Hamilton--Jacobi theory}\vskip 0.25cm
\centerline{\Large \bf  on a Nambu-Jacobi manifold}

\medskip
\medskip

\centerline{M. de Le\'on and C.
Sard\'on}
\vskip 0.5cm
\centerline{Instituto de Ciencias Matem\'aticas, Campus Cantoblanco}\vskip 0.2cm
\centerline{Consejo Superior de Investigaciones Cient\'ificas}
\vskip 0.2cm
\centerline{C/ Nicol\'as Cabrera, 13--15, 28049, Madrid. SPAIN}

\begin{abstract}
In this paper we propose a geometric Hamilton--Jacobi theory on a Nambu--Jacobi manifold.   
The advantange of a geometric Hamilton--Jacobi theory is that if a Hamiltonian vector field $X_H$
can be projected into a configuration manifold by means of a one-form $dW$, then the integral curves of the projected vector field $X_H^{dW}$ can be transformed into
integral curves of the vector field $X_H$ provided that $W$ is a solution of the Hamilton--Jacobi equation. This procedure allows us
to reduce the dynamics to a lower dimensional manifold in which we integrate the motion.
On the other hand, the interest of a Nambu--Jacobi structure resides
in its role in the description of dynamics in terms of several Hamiltonian functions. It
appears in fluid dynamics, for instance. Here, we derive an explicit expression
for a geometric Hamilton--Jacobi equation on a Nambu--Jacobi manifold and apply it to the third-order Riccati differential equation as an example.
\end{abstract}


\section{Motivation}

The Hamilton--Jacobi equation (HJ equation) constitutes the third complete formulation of classical mechanics, after Newtonian and Hamiltonian mechanics.
It is a first-order partial differential equation for a generating function $S(q^i,t)$ on a $n$-dimensional configuration manifold $Q$ with 
local canonical coordinates $\{q^i, i=1,\dots,n\}$ and $H=H(q^i,p_i)$ is the Hamiltonian function of the system on $T^{*}Q$,
which is locally coordinated by $\{q^i,p_i, i=1,\dots,n\}$. Explicitly,

\begin{equation}\label{tdepHJ}
 \frac{\partial S}{\partial t}+H\left(q^i,\frac{\partial S}{\partial q^i}\right)=0.
\end{equation}

This equation is particulary useful for the identification of conserved quantities and 
roots in the philosophy of finding a canonical transformation with generating function $S(q^i,t)$ that leads to the equilibrium
of a mechanical system \cite{Marsden,Arnold}. 
The generating function $S(q^i,t)$ is also interpreted as the action of a variational principle,
\begin{equation}
S=\int_{(q_1,t_1)}^{(q_n,t_n)} L(q(t),\dot{q}(t),t)dt
\end{equation}
such that the condition $\frac{\delta S}{\delta t}=0$ retrieves the Hamiltonian equations \cite{Gold}.

It is possible to separate the temporal dependency on $S$ through the Ansatz $S=W(q^1,\dots,q^n)-Et$,
where $E$ is the total energy of the system. This choice gives rise to the {time-independent HJ equation} \cite{Marsden,Gold},
which can be interpreted geometrically:
\begin{equation}\label{HJeq1}
 H\left({q}^i,\frac{\partial W}{\partial {q}^i}\right)=E.
\end{equation}


Concerning the geometric interpretation of a HJ theory, the primordial observation is, on a symplectic phase space, that if a Hamiltonian
vector field $X_{H}:T^{*}Q\rightarrow TT^{*}Q$ can be projected into a vector field $X_H^{dW}:Q\rightarrow TQ$ on a lower
dimensional manifold by means of a 1-form $dW$, then the integral curves of the projected
vector field $X_{H}^{dW}$ can be transformed into integral curves of $X_{H}$ provided that $W$ is a solution of \eqref{HJeq1}. If we define the projected vector field as:
\begin{equation}
 X_H^{dW}=T\pi\circ X_H\circ dW,
\end{equation}
where $T{\pi}$ is the induced projection on the tangent space, $T{\pi}:TT^{*}Q\rightarrow T^{*}Q$ by the canonical projection
$\pi:T^{*}Q\rightarrow Q$, it implies the commutativity of the diagram below:

\[
\xymatrix{ T^{*}Q
\ar[dd]^{\pi} \ar[rrr]^{X_H}&   & &TT^{*}Q\ar[dd]^{T\pi}\\
  &  & &\\
 Q\ar@/^2pc/[uu]^{dW}\ar[rrr]^{X_H^{dW}}&  & & TQ}
\]
\noindent
provided that $dW$ is a Lagrangian submanifold, since $dW$ is exact and then it is closed. This condition gave rise to the introduction of Lagrangian submanifolds in dynamics was necessary.
Lagrangian submanifolds are very important objects in Hamiltonian mechanics, since the dynamical equation (hamiltonian or lagrangian)
can be described as lagrangian submanifolds of convenient symplectic manifolds.

The pioneer in this purpose was Tulczyjew who characterized the image of local Hamiltonian vector fields on a symplectic manifold $(M,\omega)$ as lagrangian submanifolds
of a symplectic manifold $(TM,\omega^{T})$, where $\omega$ is the tangent lift of $\omega$ to $TM$ \cite{Tulczy}. This result was later generalized to Poisson manifolds \cite{GrabUrb} and Jacobi manifolds \cite{IbaLeonMarrDiego1}.

Using the approach discussed above, the HJ theory has also been extended to nonholonomic mechanics, geometric mechanics on Lie algebroids, singular systems, control theory,
classical field theories and different geometric backgrounds \cite{CGMMMLRR1,CGMMMLRR,LeonDiegoMarrSalVil,LeonIglDiego,LeonMarrDiego08,LeonMarrDiegoVaq}.
This proves the wide applicability of the geometric interpretation of the HJ theory and its recent interest among the scientific community.

In this paper, we deal with Nambu--Jacobi manifolds (NJ manifolds).
The NJ structure is a generalization of Nambu--Poisson structures (NP structures) and both appeared as an extension 
of mechanics on even and odd higher-dimensional phase spaces. In particular, both NJ and NP structures include
$n$-dimensional brackets which are very useful for descriptions of physical systems equipped with several Hamiltonian functions.
These structures have also played an important role in the study of Dirac's contraints and Nambu's mechanics \cite{Flato}.
The first bracket of order $n>2$ was the original three-dimensional Nambu bracket \cite{Nambu}, defined in terms of a Jacobian:
\begin{equation}\label{nambubracket}
 \{H_1,H_2,f\}=\frac{\partial (H_1,H_2,f)}{\partial (x,y,z)},
\end{equation}
with canonical variables satisfying $[x,y,z]=1$ and interpreted
as a bracket defined by the canonical volume form in $\mathbb{R}^3$. This bracket served as a bracket of a three-dimensional phase space
for the dynamics of particles composed by three quarks \cite{Nambu}. Also, it has showed its applicability
in noncanonical theories of perfect fluid dynamics \cite{guha,nb}. For example, consider Eulerian variables
for an incompressible fluid in 3D governed by the vorticity equation:
\begin{equation}
 \frac{\partial \Omega}{\partial t}+(u\nabla)\Omega-(\Omega \nabla)u=0,
\end{equation}
where $u=(u_1,u_2,u_3)$ is the velocity, the vorticity is $\Omega=(\Omega_1,\Omega_2,\Omega_3)$ and they are related
by $\Omega=\nabla \times u$ and $\nabla u=0$.
The total energy
\begin{equation}
 H=\frac{1}{2}\int \Omega u d^3x=-\frac{1}{2}\int \Omega A d^3x
\end{equation}
and the helicity
\begin{equation}
 h=\frac{1}{2}\int \Omega u d^3x
\end{equation}
are conserved, assuming that $u$ vanishes at infinity and $A$ is a vector potential such that $u=-\nabla\times A$, $\nabla A=0$.
The energy and helicity are Casimir functions acting like Hamiltonians for Nambu mechanics.
Now, the evolution of an arbitrary functional $F=F(\Omega)$ will be given by the Nambu bracket defined in \eqref{nambubracket}:
\begin{equation}
 \frac{dF}{dt}+\{F,h,H\}=0.
\end{equation}

Afterwards, a generalization of the three-dimensional Nambu bracket to an $n$-order bracket was provided \cite{Nambu,Takh},
as the $n$-dimensional Jacobian:


\begin{equation}\label{stebracket}
 \{H_1,\dots,H_{n-1},f\}=\frac{\partial (H_1,\dots,H_{n-1},f)}{\partial (x^1,\dots,x^n)}.
\end{equation}

Many hierarchies of differential equations are endowed with a recursion
relation that generates subsequent members of the hierarchy, conserved quantities and multiple compatible Hamiltonian functions
describing the differential equation. For example, in \cite{FStrachan} it is shown that the dispersionless Toda hierarchy 
can be described in terms of multiple compatible Hamiltonian functions,
 \begin{equation}\label{hamfunc}
  H_n=\int{h_n(u)dx},\quad h_n=(n+1)^{-1}Q_{n+1},
 \end{equation}
where $Q$ are symmetric polynomials of $u$ and depend on a zero-curvature metric. The evolution of an observable, let us say $f$, is governed by the $n$-dimensional bracket \eqref{stebracket} of the observable and Hamiltonian functions
 in \eqref{hamfunc}.

 Later, the Nambu bracket was generalized by Takhtajan \cite{Takh} by imposing the fundamental identity:
 
 \begin{align}\label{fundidentity}
\{H_1,\dots,H_{n-1},&\{g_1,\dots,g_n\}\}=\nonumber\\
&\sum_{i=1}^n \{g_1,\dots,g_{i-1},\{H_1,\dots,H_{n-1},g_i\},g_{i+1},\dots,g_n\}
\end{align}
for real valued and infinitely differentiable functions $H_1,\dots,H_{n-1},g_1,\dots,g_n.$
Expression \eqref{fundidentity} is merely an extension of the well-known Jacobi identity. This led him to the introduction of a NP manifold as a manifold
whose ring of functions is endowed with an $n$-dimensional bracket \eqref{stebracket} that satisfies \eqref{fundidentity}.
This fundamental identity was independently considered by other authors (see \cite{SW}).
Moreover, these axioms are just the ones of an $n$-Lie algebra introduced Filippov \cite{Fi} in 1985.
The extension of this bracket to an analogous of Jacobi bracket was simultaneosuly considered by several authors \cite{GM,IbaLeonMarrDiego,iba,MVV}.

The HJ theory for NP manifolds was devised by the present authors in \cite{NPLeonSardon}. Our aim here is to proceed similarly
in the case of a NJ manifold. In this way, the outline of the paper is the following: in Section 2, we recall the fundamentals
on the geometry of NJ manifolds, their structural theorem and present dynamics on such manifolds. In particular, we
focus on a particular case of NJ manifolds, that is the case of volume NJ manifolds. For the development of a geometric HJ theory,
we introduce the notion of Lagrangian submanifolds of a NJ manifold. In Section 3, we develop a geometric HJ theory on volume NP manifolds
and obtain an explicit expression for a HJ equation. 
To finish, in Section 4, we propose an application of the geometric HJ theory on volume NJ manifolds through an
example of physical interest: a four-dimensional system of Riccati first-order differential equations. We derive a volume NJ structure
for such equation and then apply the HJ theory. A particular solution for the equation is indicated.

We introduce the following notation: we will denote the algebra of $C^{\infty}$ real-valued functions as $C^{\infty}(M,\mathbb{R})$,
$\mathfrak{X}(M)$ is the $C^{\infty}(M,\mathbb{R})$ module of vector fields on $M$. The space $\mathcal{V}^{k}(M)$ is the space of $k$-vectors
on $M$ and $\Omega^k(M)$ is the space of $k$-forms on $M$.

Assume all mathematical objects to be $C^{\infty}$, globally defined and that manifolds are connected. This permits us to suit technical
details while highlightning the main aspects of the theory.

\section{Nambu--Jacobi manifolds}

Let us consider a generalized almost Jacobi manifold $(M,\{,\dots,\})$ where $M$ is an $m$-dimensional manifold and the bracket 
is a generalized almost Jacobi bracket $\{,\dots,\}:C^{\infty}(M,\mathbb{R})\times\dots ^{n)} \times C^{\infty}(M,\mathbb{R})$ of
order $n$. It has two fundamental properties:
 \begin{equation*}
  \{f_1,\dots,f_n\}=\epsilon_{\sigma}\{f_{\sigma(1)},\dots,f_{\sigma(n)}\},
 \end{equation*}
that implies that the bracket is skew-symmetric and it is a first-order linear differential operator on $M$ with respect to each argument

\begin{align*}
 \{f_1g_1,f_2\dots,f_n\}=f_1\{g_1,&f_2,\dots,f_n\} \nonumber \\
 &+g_1\{f_1,f_2,\dots,f_n\}-f_1g_1\{1,f_2,\dots,f_n\}
\end{align*}
for $f_1,\dots,f_n\ \text{and}\ g_1\in C^{\infty}(M,\mathbb{R})$. We can define the multivectors $\Lambda\in \mathcal{V}^n(M)$ and $\Box\in \mathcal{V}^{n-1}(M)$ as follows:

\begin{align*}
 &\Box(df_1,\dots,df_{n-1})=\{1,f_1,\dots,f_{n-1}\},\\
 &\Lambda(df_1,\dots,df_n)=\{f_1,\dots,f_n\}+\sum_{i=1}^n (-1)^i f_i \Box (df_1,\dots,d\widehat{f_i},\dots,df_n)
\end{align*}

Conversely, any pair $(\Lambda,\Box)\in \mathcal{V}^n(M)\oplus \mathcal{V}^{n-1}(M)$
defines a generalized almost Jacobi bracket of order $n$ on $M$, given by:

\begin{align}
 \{f_1,\dots,f_n\}=\Lambda(df_1,&\dots,df_n)\nonumber \\
 &+\sum_{i=1}^n (-1)^{i-1} f_i \Box (df_1,\dots,d\widehat{f_i},\dots,df_n)
\end{align}
\noindent
such that $(M,\{,\dots,\})$ a generalized almost Jacobi structure. If $\Box=0$, we recover a generalized almost Poisson structure. Furthermore, if we add the integrability condition

\begin{align}
 \{f_1,\dots,f_{n-1},&\{g_1,\dots,g_n\}\}\nonumber\\
 &=\sum_{i=1}^n \{g_1,\dots,g_{i-1},\{f_1,\dots,f_{n-1},g_i\},g_{i+1},\dots,g_n\},
\end{align}

\noindent
we have that the pair $(M,\{,\dots,\})$ is a Nambu--Jacobi structure, or equivalently, we say that $(M,\Lambda,\Box)$ is a Nambu--Jacobi manifold.
In the case that $\Box=0$, we have a Nambu--Poisson structure.

We define a vector field $X_{H_1,\dots,H_{n-1}}^{\Lambda}$, understood as the vector field
  on a NP manifold $(M,\Lambda)$ as:
  \begin{equation}
   X_{f_1,\dots,f_{n-1}}^{\Lambda}=\sharp_{\Lambda}(df_1\wedge\dots\wedge df_{n-1})
  \end{equation}
and define the morphism $\sharp_{\Lambda}:\Lambda^{n-1}(M)\rightarrow \mathfrak{X}(M)$ induced by $\Lambda$.
Similarly, $X_{f_1,\dots,f_{n-1}}^{\Box}$ can be understood as a vector field
  on a NP manifold $(M,\Box)$ as:
\begin{equation}
   X_{f_1,\dots \widehat{f}_i \dots f_{n-1}}^{\Box}=\sharp_{\Box}(df_1\wedge\dots\wedge d\widehat{f}_i\wedge \dots\wedge df_{n-1})
  \end{equation}
and defines the morphism $\sharp_{\Box}:\Lambda^{n-2}(M)\rightarrow \mathfrak{X}(M)$ induced by $\Box$.
The operators $(\Lambda,\Box)$ induce the following pairings, correspondingly:
\begin{equation}
 \Lambda(\alpha_1,\dots,\alpha_{n-1},\beta)=\langle \sharp_{\Lambda}(\alpha_1,\dots,\alpha_{n-1}),\beta \rangle
\end{equation}
and
\begin{equation}
 \Box(\alpha_1,\dots,\alpha_{n-2},\beta)=\langle \sharp_{\Box}(\alpha_1,\dots,\alpha_{n-2}),\beta \rangle
\end{equation}
for $\alpha_1,\dots,\alpha_{n-1},\beta\in \Omega^{1}(M).$ A vector field on a NJ manifold $(M,\Lambda,\Box)$ is defined as:
  \begin{equation}\label{defvf}
   X_{f_1,\dots,f_{n-1}}=X_{f_1,\dots,f_{n-1}}^{\Lambda}+\sum_{i=1}^{n-1} (-1)^{i-1} f_i X_{H_1,\dots,\widehat{f}_i,\dots,f_{n-1}}^{\Box}.
  \end{equation}
  
  \medskip

So, we define the characteristic distribution on a point $x\in M$ for a NJ manifold $(M,\Lambda,\Box)$ as:
\begin{equation}\label{cd}
 \mathcal{C}_x=\langle X_{f_1,\dots,f_{n-1}}(x)\rangle=\text{Im}(\sharp_{\Lambda}((x))+\text{Im}(\sharp_{\Box}(x)).
\end{equation}

Now, let us point out some properties of NJ manifolds.

\medskip
{\bf Properties:}

\medskip

 Let $(M,\Lambda,\Box)$ be a NJ manifold, $M$ has dimension $m$, $\Lambda$ is of order $n$ and $\Box$ of order $(n-1)$ and such that $n>2$. Then,
 \begin{enumerate}
  \item $(M,\Lambda)$ is a NP structure on $M$ of order $n$.
  \item $(M,\Box)$ is a NP structure on $M$ of order $n-1$.
  \item For every point $x\in M$ where $\Lambda(x)\neq 0$, there exists a one-form $\theta_x\in T^{*}M$ such that
  \begin{equation}
   \iota_{\theta_x}\Lambda(x)=\Box(x).
  \end{equation}
  \item The Lie derivative $\mathcal{L}_{X_{f_1,\dots,f_{n-2}}^{\Box}}\Lambda=0.$ 
\end{enumerate}

\subsection{Volume Nambu--Jacobi manifolds}

 A volume Nambu--Jacobi manifold $(M,\Omega,\theta)$, (henceforth VNJ manifold) is a Nambu--Jacobi manifold in which
 $\text{dim}(M)=n$, where $\Omega$ is a volume form on $M$ and $\theta$ is a one-form that is closed, $d\theta=0$. 
 There is an associated pair $(\Lambda_{\Omega},\Box_{\Omega})\in \mathcal{V}^n(M)\oplus \mathcal{V}^{n-1}(M)$ for a VNJ
 manifold, operating as follows:
  \begin{equation}
  \Lambda_{\Omega}(df_1,\dots,df_n)=\{f_1,\dots,f_n\}_{\Lambda_{\Omega}}
 \end{equation}
 and the $(n-1)$-order skew symmetric tensor $\Box_{\Omega}$ defined by:
 \begin{equation}
  \Box_{\Omega}=\iota_{\theta}\Lambda_{\Omega}
 \end{equation}
 where the bracket induced by $\Lambda_{\Omega}$ is defined by:
 \begin{equation}\label{bracketindlam}
\Omega\{f_1,\dots,f_n\}_{\Lambda_{\Omega}}=df_1\wedge \dots \wedge df_n
 \end{equation}

 A particular example is $M=\mathbb{R}^n.$ Choosing canonical coordinates $\{x^i, i=1,\dots,n\},$
 the volume form here is $\Omega_{\mathbb{R}^n}=dx^1\wedge \dots \wedge dx^n,$
 the multivectors $(\Lambda,\Box)\in \mathcal{V}^n(M)\oplus \mathcal{V}^{n-1}(M)$ take the form:
 
 \begin{equation}\label{bivec}
 \Lambda_{\mathbb{R}^n}=\frac{\partial}{\partial x^1}\wedge \dots \wedge \frac{\partial}{\partial x^n},\qquad \Box_{\mathbb{R}^n}=\frac{\partial}{\partial x^1}\wedge \dots \wedge \frac{\partial}{\partial x^{n-1}}.
\end{equation}
and $\theta$ is locally expressed as $\theta_{\mathbb{R}^n}=(-1)^{n}dx^n$.
The brackets induced by $(\Lambda_{\Omega},\Box_{\Omega})$ on a volume manifold $(M,\Omega,\theta)$ take the form:
 
 \begin{equation}
 \{f_1,\dots,f_n\}_{\Lambda_{\mathbb{R}^n}}=\frac{\partial (f_1,\dots,f_n)}{\partial (x^1,\dots,x^{n})},
\end{equation}
 and
 \begin{equation}
 \{f_1,\dots,\widehat{f}_i,\dots,f_n\}_{\Box_{\mathbb{R}^n}}=\frac{\partial (f_1,\dots,\widehat{f}_i,\dots,f_n)}{\partial (x^1,\dots,x^{n-1})}.
\end{equation}

\begin{definition}
 Consider a VNJ manifold $(M,\Omega,\theta)$ and $H_1,\dots,H_{n-1} \in C^{\infty}(M,\mathbb{R})$. Then,
 the expression in coordinates on  $\mathbb{R}^n$ for a Hamiltonian vector field on a NJ manifold \eqref{defvf} is:

{\begin{footnotesize}
\begin{align}\label{fullvf}
 X_{H_1,\dots,H_{n-1}}=&\sum_{k=1}^{n-1} (-1)^{n-k} \frac{\partial (H_1,\dots,H_{n-1})}{\partial (x^1,\dots,\widehat{x}^k,\dots,x^n)}\frac{\partial}{\partial x^k}+\frac{\partial (H_1,\dots,H_{n-1})}{\partial (x^1,\dots,x^{n-1})}\frac{\partial}{\partial x^n}\nonumber\\
 &+\sum_{i=1}^{n-1} (-1)^{i-1} H_i \frac{\partial (H_1,\dots,\widehat{H}_i,\dots,H_{n-1})}{\partial (x^1,\dots, x^{n-2})}\frac{\partial}{\partial x^{n-1}}\nonumber\\
 &+\sum_{i=1}^{n-1}\sum_{j=1}^{n-2} (-1)^{n-j+i-2} H_i \frac{\partial (H_1,\dots,\widehat{H}_i,\dots,H_{n-1})}{\partial (x^1,\dots,\widehat{x}^j,\dots, x^{n-1})}\frac{\partial}{\partial x^{j}}.
\end{align}
\end{footnotesize}}
This is a Hamiltonian system on a VNJ manifold.
\end{definition}

We can extend the model of VNJ manifolds proposed to a more general situation. If $(M,\Lambda)$ is a NP with $n>2$
 and $\theta$ is a closed one-form on $M$, then $(M,\Lambda,\iota_{\theta}\Lambda)$ is a NJ structure of order $n$ on $M$.
 Furthermore, if $\Lambda_{\Omega}$ is a tensor of order $n>2$ associated with a volume form $\Omega$ on $M$ and $\theta$ is a one-form,
  then the triple $(M,\Lambda_{\Omega}, \iota_{\theta}\Lambda_{\Omega})$ is a NJ structure \cite{iba}.
Now, let us now introduce the structural theorem of NJ manifolds \cite{iba}. The structure in leaves of the foliation encompasses
the former examples, which are canonical examples for these manifolds.
 
\begin{theorem}[Structure theorem]
 Let $(M,\Lambda,\Box)$ be a Nambu--Jacobi manifold of order $n\geq 3$ and $x\in M$. Suppose that $\mathcal{D}$ is a characteristic
 foliation of $M$ and that $L$ is the leaf of $\mathcal{D}$ passing through $x$. Then, $(\Lambda,\Box)$ reduces to a Nambu--Jacobi
 structure $(\Lambda_L,\Box_L)$ on $L$,
 \begin{enumerate}
  \item If $\Lambda(x)\neq 0$, then $L$ has dimension $n$ associated with a volume form on $L$ and there exists a closed one-form $\theta_L$
  on $L$ such that $\Box_L=\iota_{\theta_L}\Lambda_L.$
  \item If $\Lambda(x)=0$ and $\Box(x)\neq 0$, then $L$ has dimension $(n-1)$ and $\Lambda_L=0$. Moreover,
  \begin{itemize}
   \item If $n>3$, then $\Box_L$ is a Nambu--Poisson structure of order $(n-1)$ associated with a volume form on $L$.
   \item If $n=3$, then $\Box_L$ is a symplectic structure.
  \end{itemize}
  \item If $\Lambda(x)=0$ and $\Box(x)=0$, then $L=\{0\}$ and the induced Nambu--Jacobi structure is trivial.

 \end{enumerate}

\end{theorem}
\begin{proof}
 Proof of this theorem can be found in reference \cite{iba}, where it was formerly stated.
\end{proof}

%
%

\subsection{Lagrangian submanifolds}

Let $(M,\Lambda,\Box)$ be a NJ manifold. We say that a submanifold $N\subset M$ is $j$-lagrangian $\forall x\in N$, $1\leqslant j\leqslant n-1,$ if:
\begin{equation}\label{lagsub}
 \sharp_{\Lambda} \text{Ann}^j(T_xN)= \mathcal{C}_x\cap T_xN
\end{equation}
where the annihilator is defined as:
\begin{equation}\label{ann}
 \text{Ann}^j(T_xN)=\{\alpha \in \Lambda^{n-1}(T_x^{*}M) | \iota_{v_1\wedge \dots\wedge x_j}\alpha=0,\forall v_1,\dots,v_j\in T_xN\}
\end{equation}
and recall the definition of $\mathcal{C}_x$ in \eqref{cd}.
\medskip

In the particular case of a VNJ manifold, expression \eqref{lagsub} reduces to:
\begin{equation}\label{lagsub2}
 \sharp_{\Lambda} \text{Ann}^j(T_xN)=\sharp_{\Lambda} \Lambda^{n-1}(T^{*}_xM)\cap T_xN.
\end{equation}

\begin{theorem}
 Given a VNJ structure $(M,\Omega,\theta)$, every submanifold $N$ of dimension $(n-1)$ is $(n-1)$-lagrangian. No other
 lagrangian submanifolds exist.
\end{theorem}

\begin{proof}
 Recall the expressions in coordinates \eqref{bivec} for the pair $(\Lambda,\Box)$.
 Considering the definition in \eqref{lagsub}, we compute the term $\sharp_{\Lambda} \Lambda^{n-1}(T_x^{*}M)$. An element $\alpha$ of $\Lambda^{n-1}(T_x^{*}M)$
has the local expression
\begin{equation}
 \alpha=\sum_{i=1}^n\alpha_{1 \dots \hat{i}\dots n} dx^1\wedge \dots \wedge d\widehat{x^i}\wedge \dots dx^n
\end{equation}
where $d\widehat{x^i}$ stands for the omitted term $dx^i$ and $\alpha_{1\dots \hat{i}\dots n}$ is the coefficient of 
the $(n-1)$-order form $\alpha\in \Lambda^{n-1}(T_x^{*}M)$ associated with the combination $dx^1\wedge \dots \wedge d\widehat{x^i}\wedge \dots dx^n$.
Then,
\begin{equation}
 \text{Im}(\alpha)=\sum_{i=1}^n (-1)^{n-i} \alpha_{1\dots \hat{i}\dots n}\frac{\partial}{\partial x^{i}}
\end{equation}
and therefore,
\begin{equation*}
\sharp_{\Lambda} \Lambda^{n-1}(T_x^{*}M)=\langle \frac{\partial}{\partial x^1},\dots,\frac{\partial}{\partial x^n}\rangle.
\end{equation*}
This implies that $\sharp_{\Lambda} \Lambda^{n-1}(T_x^{*}M)=T_xM=\langle \frac{\partial}{\partial x^1}, \dots, \frac{\partial}{\partial x^n}\rangle$
and hence, the right-hand side of expression \eqref{lagsub} in this particular case equals
\begin{equation*}
 \sharp_{\Lambda} \Lambda^{n-1}(T_x^{*}M)\cap T_xN=T_xN
\end{equation*}
Computing the left-hand side term of \eqref{lagsub}, let us assume that $N$ is $(n-1)$-dimensional and the equations
of the submanifold are $\{x^n=0\}$ in the coordinates we gave chosen. So, the tangent space to this submanifold
is generated by $TN=\langle \frac{\partial}{\partial x^1}, \dots, \frac{\partial}{\partial x^{n-1}}\rangle$.
The elements of the annihilator when $j=n-1$ are $\alpha\in \Lambda^{n-1}(T^{*}_xM)$ such that
\begin{equation}
 \iota_{\frac{\partial}{\partial x^1}\wedge \dots \wedge \frac{\partial}{\partial x^{n-1}}}\alpha=0,\quad \forall \frac{\partial}{\partial x^1},\dots,\frac{\partial}{\partial x^{n-1}}\in T_xN
\end{equation}
The $\alpha$'s fulfilling this condition are of the form
\begin{equation}
 \text{Ann}^{n-1}(TN)=\sum_{i=1}^{n-1}\alpha_{1\dots\hat{i}\dots n} dx^1\wedge \dots d\widehat{x^i}\dots dx^n.
\end{equation}
Hence, the annihilator of orfer $j-1$ for a submanifold $N$ of dimension $n-1$ is generated by
\begin{equation}
 \text{Ann}^{n-1}(TN)=\langle \frac{\partial}{\partial x^1},\dots,\frac{\partial}{\partial x^{n-1}}\rangle=TN
\end{equation}
This means that the left-hand side and the right-hand side of \eqref{lagsub} are equal, and the submanifold $N\subset M$ of
dimension $n-1$ is a Lagrangian submanifold of order $n-1$.
It is easy to see that $\text{Ann}^{j\neq n-1}(T_xN)=\{0\}$ and that for submanifolds $N$ of dimension different from $n-1$
the annihilator of any orfer $j$ equals zero, namely,  $\text{Ann}^{j}(T_xN)=\{0\}, \text{for}\ \text{dim}(M)\neq n-1.$

\end{proof}

\section{Hamilton-Jacobi theory on VNJ manifolds}

Consider a VNJ structure $(M,\Omega,\theta)$, a fibration $\pi:M\rightarrow N$ such that $\text{dim}(M)=n$ and $\text{dim}(N)=n-1$
and a section $\gamma:N\rightarrow M$. 
The projected vector field on $N$ can be defined as:
\begin{equation}\label{defcom}
 X_{H_1,\dots,H_{n-1}}^{\gamma}=T_{\pi}\circ X_{H_1,\dots,H_{n-1}}\circ \gamma.
\end{equation}

We depict it with a diagram:

\[
\xymatrix{ M
\ar[dd]^{\pi} \ar[rrr]^{X_{H_1,\dots,H_{n-1}}}&   & &TM\ar[dd]^{T\pi}\\
  &  & &\\
 N\ar@/^2pc/[uu]^{\gamma}\ar[rrr]^{X_{H_1,\dots,H_{n-1}}^{\gamma^n}}&  & & TN}
\]
\noindent
Recall from former sections, the Hamiltonian vector field on a VNJ in adapted 
coordinates can be locally written as in \eqref{fullvf}. Then, the following theorem results:


\begin{theorem}
 The two vector fields $X_{H_1,\dots,H_{n-1}}^{\gamma}$ and $X_{H_1,\dots,H_{n-1}}$ are $\gamma$-related if and only if
 the following equation is satisfied:
 
 \begin{align}\label{hjenj}
\sum_{i=1}^{n-1} (-1)^{i-1} H_i& \left(d(H_1\circ \gamma)\wedge \dots \wedge \widehat{d(H_i\circ \gamma)}\wedge \dots \wedge d(H_{n-1}\circ \gamma)\right)\nonumber \\
&+d(H_1\circ \gamma)\wedge \dots \wedge d(H_{n-1}\circ \gamma)=0. 
\end{align}
 
\end{theorem}

\begin{proof}
Referring to the coordinate expressions for a VNJ structure \eqref{bivec},  we look
for a section $\gamma=\gamma(x^1,\dots,x^{n-1},\gamma^{n}(x^1,\dots,x^{n-1}))$ such that \eqref{defcom} is fulfilled. The projection 
of the vector field $X_{H_1,\dots,H_{n-1}}$ on $M$ \eqref{fullvf} onto $N$ reads:

{\begin{footnotesize}
\begin{align}
X_{H_1,\dots,H_{n-1}}^{\gamma}=&\sum_{k=1}^{n-1} (-1)^{n-k} \frac{\partial (H_1,\dots,H_{n-1})}{\partial (x^1,\dots,\widehat{x}^k,\dots,x^n)}\frac{\partial}{\partial x^k}\nonumber\\
 &+\sum_{i=1}^{n-1} (-1)^{i-1} H_i \frac{\partial (H_1,\dots,\widehat{H}_i,\dots,H_{n-1})}{\partial (x^1,\dots, x^{n-2})}\frac{\partial}{\partial x^{n-1}}\nonumber\\
 &+\sum_{i=1}^{n-1}\sum_{j=1}^{n-2} (-1)^{n-j+i-2} H_i \frac{\partial (H_1,\dots,\widehat{H}_i,\dots,H_{n-1})}{\partial (x^1,\dots,\widehat{x}^j,\dots, x^{n-1})}\frac{\partial}{\partial x^{j}}
\end{align}
\end{footnotesize}}

and its tangent image by the section $\gamma$ is:
{\begin{footnotesize}
\begin{align}\label{tangentfull}
 T{\gamma}&X_{H_1,\dots,H_{n-1}}^{\gamma}=\sum_{k=1}^{n-1} (-1)^{n-k} \frac{\partial (H_1,\dots,H_{n-1})}{\partial (x^1,\dots,\widehat{x}^k,\dots,x^n)}\left(\frac{\partial}{\partial x^k}+\frac{\partial \gamma^n}{\partial x^k}\frac{\partial}{\partial x^n}\right)\nonumber\\
 &+\sum_{i=1}^{n-1} (-1)^{i-1} H_i \frac{\partial (H_1,\dots,\widehat{H}_i,\dots,H_{n-1})}{\partial (x^1,\dots, x^{n-2})}\left(\frac{\partial}{\partial x^{n-1}}+\frac{\partial \gamma^n}{\partial x^{n-1}}\frac{\partial}{\partial x^{n}}\right)\nonumber\\
 &+\sum_{i=1}^{n-1}\sum_{j=1}^{n-2} (-1)^{n-j+i-2} H_i \frac{\partial (H_1,\dots,\widehat{H}_i,\dots,H_{n-1})}{\partial (x^1,\dots,\widehat{x}^j,\dots, x^{n-1})}\left(\frac{\partial}{\partial x^{j}}+\frac{\partial \gamma^n}{\partial x^{j}}\frac{\partial}{\partial x^{n}}\right).
\end{align}
\end{footnotesize}}

By direct comparison of \eqref{fullvf} and \eqref{tangentfull}, we obtain the following equation:

{\begin{footnotesize}
\begin{align}
 \sum_{k=1}^{n-1} &(-1)^{n-k}\frac{\partial (H_1,\dots,H_{n-1})}{\partial (x^1,\dots,\widehat{x}^k,\dots,x^n)}\frac{\partial \gamma^n}{\partial x^k}-\frac{\partial (H_1,\dots,H_{n-1})}{\partial (x^1,\dots,x^{n-1})}\nonumber\\
+&\sum_{k=1}^{n-1} (-1)^{i-1} H_i \frac{\partial (H_1,\dots,\widehat{H}_i,\dots,H_{n-1})}{\partial (x^1,\dots,x^{n-2})}\frac{\partial \gamma^n}{\partial x^{n-1}}\nonumber\\
+&\sum_{i=1}^{n-1}\sum_{j=1}^{n-2} (-1)^{n-j+i-2} H_i \frac{\partial (H_1,\dots,\widehat{H}_i,\dots,H_{n-1})}{\partial (x^1,\dots,\widehat{x}^j,\dots,x^{n-1})}\frac{\partial \gamma^n}{\partial x^j}=0,
 \end{align}
 \end{footnotesize}}

which can be rewritten in the following compact form
{\begin{footnotesize}
\begin{align}
\sum_{i=1}^{n-1} (-1)^{i-1} H_i& \left(d(H_1\circ \gamma)\wedge \dots \wedge \widehat{d(H_i\circ \gamma)}\wedge \dots \wedge d(H_{n-1}\circ \gamma)\right)\nonumber \\
&+d(H_1\circ \gamma)\wedge \dots \wedge d(H_{n-1}\circ \gamma)=0,
\end{align}
\end{footnotesize}}
which is precisely the equation given in \eqref{hjenj}.

\end{proof}

Expression \eqref{hjenj} receives the name of {Hamilton--Jacobi equation on a volume Nambu--Jacobi manifold}.
We say that $\gamma$ is a {solution of the Hamilton--Jacobi problem on a volume Nambu--Jacobi manifold}
for $(M,\Omega,\theta)$.

Next, we consider a general Nambu--Jacobi manifold $(M,\Lambda,\Box)$ where the dimension of $M$ is $m$ and the order
of the multivector $\Lambda$ is $n$. Consider a fibration $\pi:M\rightarrow N$ 
over an $n$-dimensional manifold $N$. Take $\gamma:N\rightarrow M$ a section of $\pi$ such that $\pi\circ \gamma=\text{Id}_N$
and  $\gamma(N)$ is a Lagrangian submanifold of $(M,\Lambda,\Box)$.
An additional hypothesis is that $\gamma(N)$ has a clean intersection with the leaves of the characteristic
foliation $\mathcal{C}$ defined by $\Lambda$. We recall that this implies that for each leaf $L\in \mathcal{C}$,
\begin{enumerate}
 \item $\gamma(N)\cap L$ is a submanifold
\item $\text{T}(\gamma(N)\cap L)=\text{T}\gamma(N)\cap \text{T}N$
\end{enumerate}
If we are assuming that $\gamma(N)$ is a $j$-Lagrangian submanifold of $(M,\Lambda,\Box)$,
then $\gamma(N)\cap L$ is a $j$-Lagrangian submanifold of $L$ with the restricted Nambu--Jacobi structure, 
that is a volume structure.
Consequently, $j=n-1$ and $N$ has dimension $n-1$.
Now, let $H_1,\dots,H_{n-1}$ be Hamiltonian functions in $M$ and $X_{H_1 \dots H_{n-1}}$ the corresponding Hamiltonian vector field.

We define the vector field on $N$,
\begin{equation*}
 X_{H_1\dots H_{n-1}}^{\gamma}=T{\pi}\circ X_{H_1 \dots H_{n-1}}\circ \gamma.
\end{equation*}
Since every Hamiltonian vector field is tangent to the characteristic foliation, one can conclude that
\begin{theorem}
 The vector fields $X_{H_1\dots H_{n-1}}^{\gamma}$ and $X_{H_1 \dots H_{n-1}}$ are $\gamma$-related if and only if
\begin{align}\label{HJeq3}
 \sum_{i=1}^{n-1} (-1)^{i-1} H_i& \left(d(H_1\circ \gamma)\wedge \dots \wedge \widehat{d(H_i\circ \gamma)}\wedge \dots \wedge d(H_{n-1}\circ \gamma)\right)\nonumber \\
&+d(H_1\circ \gamma)\wedge \dots \wedge d(H_{n-1}\circ \gamma)=0. 
\end{align}

\end{theorem}
Therefore, \eqref{HJeq3} will be called the {HJ equation for a general Nambu--Jacobi manifold} and $\gamma$
satisfying \eqref{HJeq3} will be a {solution of the HJ problem on a general Nambu--Jacobi manifold} for $(M,\Lambda,\Box)$
with a multivector $\Lambda$ of order $n$ and the dimension of the manifold $M$ is $m$. 

\section{Application}

%


Let us consider a four-dimensional system of first-order Riccati differential equations \cite{lucas,riccati,sardon}. It is a system of four copies of the first-order 
Riccati differential equation on $\mathcal{O}=\{(x_1,x_2,x_3,x_4)\,|\, x_i\neq x_j,i\neq j=1,\ldots,4\}\subset \mathbb{R}^4,$ given by
\begin{equation}\label{CoupledRic}
\frac{{\rm d} x^i}{{\rm d} t}=a_0(t)+a_1(t)x^i+a_2(t)(x^i)^2,\qquad i=1,\ldots,4,\\
\end{equation}
where $a_0(t),a_1(t),a_2(t)$ are arbitrary $t$-dependent functions. More generally, the term Riccati equation is used to refer to matrix equations with an analogous quadratic term,
which occur in both continuous-time and discrete-time linear-quadratic-Gaussian control. The steady-state (non-dynamic) version of these is referred to as the algebraic Riccati equation \cite{lucas,riccati,sardon}.

These equations are compatible with several Hamiltonian structures, namely, the triples 
$(\mathcal{O},\omega_l,H_{l})$ such that $H_l$ is a Hamiltonian function with respect to a
presymplectic form $\omega_l$.  

For a VNJ structure, let us consider three of them for values $l=2,3,4$ and the values of $k$ range from 1 to 4.

It reads:

{\begin{footnotesize}
$$(\mathcal{O},\omega_l,H_{l})=
\begin{cases}
\quad \mathcal{O}=\{(x_1,x_2,x_3,x_4)\,|\, x_i\neq x_j,i\neq j=1,\ldots,4\}\subset \mathbb{R}^4,\\
 \quad \omega_l=\sum_{k<l}\frac{dx^k\wedge dx^l}{(x^k-x^l)^2}+\sum_{k>l}\frac{dx^l\wedge dx^k}{(x^l-x^k)^2}\\
\quad H_l=\left(\sum_{k<l} \frac{x^kx^l}{x^k-x^l}+\sum_{k>l} \frac{x^lx^k}{x^l-x^k}\right)+b_1(t) \left(\sum_{k<l} \frac{1}{x^k-x^l}+\sum_{k>l} \frac{1}{x^l-x^k}\right)
 \end{cases}
$$
\end{footnotesize}}

To obtain a VNJ structure, we need to find a volume form $\Omega$ and a one-form $\theta$ compatible
with $(\mathcal{O},H_1,H_2,H_3)$.
The procedure consists of retrieving \eqref{CoupledRic} from the Nambu-Jacobi brackets
{\begin{footnotesize}
\begin{equation}\label{exnjbracket}
\dot{x}^i=\{H_1, H_2, H_3, x^i\},\quad \forall i=1,2,3,4.
\end{equation}
\end{footnotesize}}
characterizing the evolution of the curves $(x^i(t)), i=1,2,3,4$. According to the theory of NJ structures, 
the computation of the bracket \eqref{exnjbracket} according to \eqref{bracketindlam} is

{\begin{footnotesize}
\begin{equation}\label{defjacobian}
 \dot{x}^i=\{H_1,H_2,H_3,x^i\}=(-1)^{4+i}\sum_{\sigma}(-1)^{\sigma_{({i_1},{i_2},{i_3})}}\prod_{\sigma(i_j),k=1,2,3}\frac{\partial H_l}{\partial x^{\sigma(i_j)}}
                       \end{equation} \end{footnotesize}}
\noindent
for $i=1,2,3,4$, $j=1,2,3$, $k=1,2,3$ and $\sigma_{({i_1},{i_2},{i_3})}$ denotes all the permutations of elements $\{x^{i_1},x^{i_2},x^{i_3}\}$ when
$(i_1,i_2,i_3)$ take values from 1 to 4 and they are simultaneously different from $i$, namely, $i_1\neq i_2\neq i_3\neq i$.

The argument takes the form:
%
%
\medskip

{\begin{footnotesize}

\begin{equation*}
\frac{\partial H_l}{\partial x^j}=F_{lj}\left(a_0(t)+a_1(t)x^l+a_2(t)(x^l)^2\right),
\end{equation*}

\begin{equation*}
\frac{\partial H_l}{\partial x^l}=F_{ll}\left(a_0(t)+a_1(t)x^k+a_2(t)(x^k)^2\right)
\end{equation*}

\end{footnotesize}}

for $k=1,2,3,4$ and $k\neq l$ and
{\begin{footnotesize}
$$F_{k,\sigma(i_j)}=
\begin{cases}
\quad F_{lj}=\frac{1}{(x^l-x^j)^2},\quad \text{when} \quad \sigma(i_j)=j\neq l,\\
 \quad F_{ll}=\left(\sum_{k<l}\frac{1}{(x^k-x^l)^2}-\sum_{k>l}\frac{1}{(x^l-x^k)^2}\right),\quad \text{when}\quad \sigma(i_j)=l.
\end{cases}\label{casesf}
$$
\end{footnotesize}}

By direct comparison between \eqref{CoupledRic} and \eqref{defjacobian}, we see there is difference of one factor,
that implies that the canonical volume form $\Omega=dx^1\wedge dx^2\wedge dx^3\wedge dx^4$ for a four-dimensional system of first-order
Riccati differential equations on a VNJ manifold is has to be conformally transformed to
{\begin{footnotesize}
$$\Omega=
\begin{cases}
F_{lj}dx^1\wedge dx^2\wedge dx^3\wedge dx^4 \quad \text{when}\quad \sigma(i_j)=j\neq l\\
F_{ll}dx^1\wedge dx^2\wedge dx^3\wedge dx^4 \quad \text{when}\quad \sigma(i_j)=l.
 \end{cases}
$$
\end{footnotesize}}

%

For this particular case, the canonical bracket $\{x^1,x^2,x^3,x^4\}=1$ turns out in
{\begin{footnotesize}
\begin{equation*}
 \{x^1,x^2,x^3,x^4\}=\frac{1}{F_{lj}},\quad  \{x^1,x^2,x^3,x^4\}=\frac{1}{F_{ll}}
\end{equation*}
\end{footnotesize}}
if $\sigma(i_j)=j\neq l$ and $\sigma(i_j)=l$, correspondingly. 
In this way, the canonical structure in \eqref{bivec} for $(\Lambda,\Box)$ when $n=4$ has to be conformally transformed into $(\bar{\Lambda},\bar{\Box})$ 
that takes the expression
{\begin{footnotesize}
$$(\bar{\Lambda},\bar{\Box})=
\begin{cases}
 (F_{lj}\Lambda,F_{lj}\Box), \quad \text{for} \quad \sigma(i_j)=j\neq l \\
 (F_{ll}\Lambda,F_{ll}\Box),\quad \text{for}\quad \sigma(i_j)=l.
\end{cases}
$$
\end{footnotesize}}
This choice leaves the one-form $\theta=dx^4$ invariant.
Hence, the VNJ structure for a four-dimensional system of first-order differential Riccati equations is $(\mathcal{O},\bar{\Omega},\theta)$, as defined above.

Let us now apply the HJ theory on this VNJ manifold. Recall that the definition of a Hamiltonian vector field on a NJ manifold is:
\begin{equation}
 X_{H_1,\dots,H_{n-1}}=X_{H_1,\dots,H_{n-1}}^{\Lambda}+\sum_{i=1}^{n-1}(-1)^{i-1} H_iX_{H_1,\dots,H_{n-1}}^{\Box}
\end{equation}
Applied to our particular case of a four-dimensional system of first-order Riccati differential equations, this expression is simplified to:
{\begin{footnotesize}
\begin{align*}
 X_{H_1 H_2 H_3}=\sharp_{\Lambda} (dH_1\wedge &dH_2\wedge dH_3)+H_1\Box(dH_2\wedge dH_3)\\
 &-H_2\Box(dH_1\wedge dH_3)+H_3\Box (dH_1\wedge dH_2)
\end{align*}
\end{footnotesize}}
This can be rewritten in following compact form:
{\begin{footnotesize}
\begin{align}\label{vf}
 X_{H_1 H_2 H_3}=\sum_{l=1}^3 &\left((-1)^{i+1}H_i\Lambda_{12}\right)_l (H_j\wedge H_k)_l \frac{\partial}{\partial x^3}+\sum_{l=1}^3 \left((-1)^{i+1}H_i\Lambda_{31}\right)_l (H_j\wedge H_k)_l \frac{\partial}{\partial x^2}\nonumber\\
 +&\sum_{l=1}^3 \left((-1)^{i+1}H_i\Lambda_{23}\right)_l (H_j\wedge H_k)_l \frac{\partial}{\partial x^1}+\left(\frac{\partial H_1}{\partial x^1}\frac{\partial H_2}{\partial x^2}\frac{\partial H_3}{\partial x^3}\right)\frac{\partial}{\partial x^4}
\end{align}
\end{footnotesize}}
\noindent
where $\Lambda_{mn}=\frac{\partial}{\partial x^m}\otimes \frac{\partial}{\partial x^n}$ is an operator acting 
on the wedge product of two functions $H_j\wedge H_k=H_j\otimes H_k-H_k\otimes H_j$, each element of the scalar product in each
entry of the operator. Notice that $j<k$, $j=1,2,3$ and $i\neq j\neq k$ for each summand.

We define a projected vector field $X_{H_1 H_2 H_3}^{\gamma}$ as:

\begin{equation}
 X_{H_1 H_2 H_3}^{\gamma}=T{\pi}\circ X_{H_1 H_2 H_3}\circ \gamma
\end{equation}

such that the diagram below is commutative:

\[
\xymatrix{ \mathcal{O}
\ar[dd]^{\pi} \ar[rrr]^{X_{H_1 H_2 H_3}}&   & &T\mathcal{O}\ar[dd]^{T\pi}\\
  &  & &\\
 \mathcal{O}|_{3}\ar@/^2pc/[uu]^{\gamma}\ar[rrr]^{X_{H_1 H_2 H_3}^{\gamma}}&  & & T\mathcal{O}|_{3}}
\]
\noindent
Let us choose a section $\gamma$ that in coordinates takes the expression $\gamma=(x^1,x^2,x^3,\gamma^{4}(x^1,x^2,x^3))$
and we denote by $\mathcal{O}_3$ the restriction of the four-dimensional space
$\{(x_1,x_2,x_3,x_4)\,|\, x_i\neq x_j,i\neq j=1,\ldots,4\}\subset \mathbb{R}^4$ to 
$\{(x_1,x_2,x_3)\,|\, x_i\neq x_j,i\neq j=1,\ldots,3\}\subset \mathbb{R}^4$. So, the projected vector field
in coordinates reads

{\begin{footnotesize}
\begin{align}
 X_{H_1 H_2 H_3}=&\sum_{l=1}^3 \left((-1)^{i+1}H_i\Lambda_{12}\right)_l (H_j\wedge H_k)_l \frac{\partial}{\partial x^3}+\sum_{l=1}^3 \left((-1)^{i+1}H_i\Lambda_{31}\right)_l (H_j\wedge H_k)_l \frac{\partial}{\partial x^2}\nonumber\\
 +&\sum_{l=1}^3 \left((-1)^{i+1}H_i\Lambda_{23}\right)_l (H_j\wedge H_k)_l \frac{\partial}{\partial x^1}
\end{align}
\end{footnotesize}}

The image of this projected vector field by $T\gamma$ is:

{\begin{footnotesize}
\begin{align}
 T{\gamma} X_{H_1 H_2 H_3}&=\sum_{l=1}^3 \left((-1)^{i+1}H_i\Lambda_{12}\right)_l (H_j\wedge H_k)_l \left(\frac{\partial}{\partial x^3}+\frac{\partial \gamma^4}{\partial x^3}\frac{\partial}{\partial x^4}\right)\nonumber\\
 &+\sum_{l=1}^3 \left((-1)^{i+1}H_i\Lambda_{31}\right)_l (H_j\wedge H_k)_l \left(\frac{\partial}{\partial x^2}+\frac{\partial \gamma^4}{\partial x^2}\frac{\partial}{\partial x^4}\right)\nonumber\\
 &+\sum_{l=1}^3 \left((-1)^{i+1}H_i\Lambda_{23}\right)_l (H_j\wedge H_k)_l \left(\frac{\partial}{\partial x^1}+\frac{\partial \gamma^4}{\partial x^1}\frac{\partial}{\partial x^4}\right)
\end{align}
\end{footnotesize}}
\noindent
that compared with \eqref{vf} provides the Hamilton--Jacobi equation for a four-dimensional system of first-order Riccati differential equations
{\begin{footnotesize}
\begin{align}
 &\sum_{l=1}^3 \left((-1)^{i+1}H_i\Lambda_{12}\right)_l (H_j\wedge H_k)_l \frac{\partial \gamma^4}{\partial x^3}+\sum_{l=1}^3 \left((-1)^{i+1}H_i\Lambda_{31}\right)_l (H_j\wedge H_k)_l \frac{\partial \gamma^4}{\partial x^2}\nonumber\\
 &+\sum_{l=1}^3 \left((-1)^{i+1}H_i\Lambda_{23}\right)_l (H_j\wedge H_k)_l \frac{\partial \gamma^4}{\partial x^1}=\frac{\partial H_1}{\partial x^1}\frac{\partial H_2}{\partial x^2}\frac{\partial H_3}{\partial x^3}.
\end{align}
\end{footnotesize}}
This equation is a quasi-linear equation that can be solved with the method of characteristics. Due to the volume of calculations,
we just leave the equation indicated.

\section*{Acknowledgements}
This work has been partially supported by MINECO MTM 2013-42-870-P and
the ICMAT Severo Ochoa project SEV-2011-0087. We kindly acknowledge the committee for choosing our contribution to the Proceeding's book
of the XXV International Workshop on Geometry and Physics (CSIC-Madrid, Spain). We thank Partha Guha for the suggestion of a Hamilton--Jacobi
theory for Nambu--Jacobi manifolds.

\end{document}